\documentclass[twocolumn,notitlepage,showpacs]{revtex4}
\usepackage{graphicx}
\usepackage{tikz-cd}
\usepackage{dcolumn}
\usepackage{amsmath}
\usepackage{amssymb}
\usepackage{mathtools}
\usepackage{hyperref}

\newtheorem{theorem}{Theorem}[section]

\newtheorem{definition}[theorem]{Definition}
\newenvironment{proof}[1][Proof]{\begin{trivlist}
\item[\hskip \labelsep {\bfseries #1}]}{\end{trivlist}}

\newcommand{\qed}{\nobreak \ifvmode \relax \else
	\ifdim\lastskip<1.5em \hskip- \lastskip
	\hskip 0.5em plus0em minus0.5em \fi \nobreak
	\vrule height0.75em width0.5em depth0.25em\fi}

\begin{document}

\title{On the existence of marginally trapped tubes in spacetimes with local rotational symmetry}

\author{Abbas M. \surname{Sherif}}
\email{abbasmsherif25@gmail.com}
\affiliation{Cosmology and Gravity Group, Department of Mathematics and Applied Mathematics, University of Cape Town, Rondebosch 7701, South Africa}

\begin{abstract}
Let $M$ be a locally rotationally symmetric spacetime with at least one of the rotation or spatial twist being non-zero. It is proved that $M$ cannot admit a non-minimal marginally trapped tube of the form $\chi=X(t)$.
\end{abstract}
\maketitle

\section{introduction}

\label{intro}

Black holes have become a topical research area over the last decade. Different numerical and analytical schemes have been setup to study the local and global dynamics of these objects, both in the astrophysical and in more abstract contexts \cite{hay3,ash1,ash2,ash3,ash7,ash8,ash9,sen1,sen3,sen5,ib1,ib2,ib3,ib4}. The local dynamics of these objects can be tracked by the evolution of their associated horizons, the so-called dynamical horizon (and more generally, marginally trappe tubes) developed by Ashtekar and others (see the references \cite{ash1,ash2,ash3}). Their existence and uniqueness in some general setting was extensively dealt with in \cite{ash3} by Ashtekar and Galloway, which established constraints on the locations and occurence of dynamical horizons.

Indeed, relating the causal character of marginally trapped tubes to explicit physical/geometrical quantities in the spacetimes is of huge interest to both mathematical and numerical relativists. A relatively recent appproach to dealing with spacetimes admitting a unit tangent directing and some preferred spatial directiion normal to the tangent one, seems a good candidate for such project. This so-called \(1+1+2\) formalism \cite{pg1,cc1,cc2,cc3} characterizes a spacetime in a covariant manner just as the \(1+3\) formalism, but with new scalar, vector and tensor quantities defined on the \(2\)-space. The first use of this formalism to analyze black hole horizons was by Ellis \textit{et al.} \cite{rit1}, where considerations were given to an astrophysical collapse scenario. In \cite{shef1} and \cite{shef2}, Sherif and coauthors extended the results in \cite{rit1} to more general spacetimes, both using an original approach, as well as adapting an approach by Booth \textit{et al.} (see \cite{ib3,ib4} and references therein), to the \(1+1+2\) covariant variables. A series of established results were obtained in a relatively easy manner, as well as a couple of original results like a classifications scheme of marginally trapped tubes up to diffeomorphisms and causal characters. For specific examples, the authors considered spacetimes in the class of locally rotationally symmetric spacetimes with no rotation and no spatial twist.

In this paper however, the class of locally rotationally symmetric spacetimes where at least one of the rotation or spatial twist is non-vanishing is considered. It is shown that for these spacetimes, any marginally trapped tube of the form \(\chi=X(t)\), will have both null expansion scalars (ingoing and outgoing) vanishing on them.

The paper is structured as follows: in Section \ref{2}, the \(1+1+2\) spacetime decomposition method employed in this work is briefly introduced, and provide some definitions of utility to this paper is provided. Section \ref{3} states and prove the main result of the paper, and then in Section \ref{4} the paper is concluded with discussion of the result.

\section{Preliminaries}\label{2}

Let (\(M,g_{ab}\)) be a \(4\)-dimensional spacetime, and to any timelike congruence, let there be associated a unit vector field \(u^a\) tangent to the congruence with \(u^au_a=-1\). Given any \(4\)-vector \(U^a\) in the spacetime, one may split \(U^a\) as
\begin{eqnarray*}
U^a&=&Uu^a + U^{\langle a \rangle },
\end{eqnarray*}
where \(U\) is the scalar along \(u^a\) and \(U^{\langle a \rangle }\) is the projected trace-free \(3\)-vector \cite{ggff2}. The splitting splits \(g_{ab}\) into components associated with the \(u^a\) and spatial directions as

\begin{eqnarray*}
h_a^{\ b}\equiv g_a^{\ b}+u_au^b,
\end{eqnarray*}
where \(h_{ab}\) project any \(3\)-vector to the \(3\)-space. This naturally gives rise to two derivatives:
\begin{itemize}
\item The \textit{covariant time derivative} (which we shall henceforth call the dot derivative)  along the observers' congruence. Given any tensor \(S^{a..b}_{\ \ \ \ c..d}\), we have \(\dot{S}^{a..b}_{\ \ \ \ c..d}\equiv u^e\nabla_eS^{a..b}_{\ \ \ \ c..d}\).

\item Fully orthogonally \textit{projected covariant derivative} \(D\) with the tensor \(h_{ab}\), with the total projection carried out on all the free indices. Given any tensor \(S^{a..b}_{\ \ \ \ c..d}\), we have \(D_eS^{a..b}_{\ \ \ \ c..d}\equiv h^a_{\ f}h^p_{\ c}...h^b_{\ g}h^q_{\ d}h^r_{\ e}\nabla_rS^{f..g}_{\ \ \ \ p..q}\).
\end{itemize}
This \(1+3\) splitting of then splits the covariant derivative of \(u^a\) as

\begin{eqnarray}\label{mmmn}
\nabla_au_b=-u_a\mathcal{A}_b+\frac{1}{3}h_{ab}\Theta+\sigma_{ab},
\end{eqnarray}
where we have defined \(\mathcal{A}_a=\dot{u}_a\) as the acceleration vector, \(\Theta\equiv D_au^a\) as the expansion and \(\sigma_{ab}=D_{\langle b}u_{a\rangle}\) as the shear tensor. 

The energy momentum tensor decompose as

\begin{eqnarray}
T_{ab}=\rho u_au_b + 2q_{(a}u_{b)} +ph_{ab} + \pi_{ab},
\end{eqnarray}
where \(\rho\equiv T_{ab}u^au^b\) defines the energy density, the heat flux is given by \(q_a=-h_a^{\ c}T_{cd}u^d\), \(p\equiv\left(1/3\right)h^{ab}T_{ab}\) is the isotropic pressure and \(\pi_{ab}\) is the anisotropic stress tensor.

Suppose there exists a spatial direction \(e^a\) (with \(e_ae^a=1\)) which is orthogonal to \(u^a\). The metric \(g_{ab}\) can be split into terms along the \(u^a\) and \(e^a\) directions (the vector field \(e^a\) splits the \(3\)-space), as well as on the \(2\)-surface, i.e. 

\begin{eqnarray}
g_{ab}=N_{ab}-u_au_b+e_ae_b,
\end{eqnarray} 
where the tensor \(N_{ab}\) projects any two vector orthogonal to \(u^a\) and \(e^a\) onto the \(2\)-surface called the sheet (\(N^{\ \ a}_a=2, u^aN_{ab}=0, \ e^aN_{ab}=0\)). This is the \(1+1+2\) splitting. This gives rise to the further splitting of the covariant derivatives along the \(e^a\) direction and on the \(2\)-surface:
\begin{itemize}
\item The \textit{hat derivative} is the spatial derivative along the vector field \(e^a\): Given a \(3\)-tensor \(\psi_{a..b}^{\ \ \ \ c..d}\), we have that \(\hat{\psi}_{a..b}^{\ \ \ \ c..d}\equiv e^fD_f\psi_{a..b}^{\ \ \ \ c..d}\).

\item The \textit{delta derivative} is the projected spatial derivative on the \(2\)-sheet (projection by the tensor \(N_a^{\ b}\)), and the projection is carried out on all the free indices: Given any \(3\)-tensor \(\psi_{a..b}^{\ \ \ \ c..d}\), we have the delta derivative as \(\delta_e\psi_{a..b}^{\ \ \ \ c..d}\equiv N_a^{\ f}..N_b^{\ g}N_h^{\ c}..N_i^{\ d}N_e^{\ j}D_j\psi_{f..g}^{\ \ \ \ h..i}\).
\end{itemize} 
Further details for the splitting procedure and all of the resulting field equations can be found in \cite{cc1}.

Let us now define locally rotationally symmetric spacetimes.

\begin{definition}
A \textbf{locally rotationally symmetric (LRS)} spacetime is a spacetime in which at each point \(p\in M\), there exists a continuous isotropy group generating a multiply transitive isometry group on \(M\) \cite{ggff2}. The general metric of LRS spacetimes is given by

\begin{eqnarray}\label{jan29191}
\begin{split}
ds^2&=-A^2dt^2 + B^2d\chi^2 + F^2 dy^2 \\
&+ \left[\left(F\bar{D}\right)^2+ \left(Bh\right)^2 - \left(Ag\right)^2\right]dz^2\\
&+ \left(A^2gdt - B^2hd\chi\right)dz,
\end{split}
\end{eqnarray}
where \(A^2,B^2,F^2\) are functions of \(t\) and \(\chi\), \(\bar{D}^2\) is a function of \(y\) and \(k\) (\(k\) fixes the geometry of the \(2\)-surfaces), and \(g,h\) are functions of \(y\). 
\end{definition}
In the limiting case that \(g=h=0\) we recover the well known spherically symmetric LRS II class of spacetimes which generalizes spherically symmetric solutions to the Einstein field equations. 

The complete set of \(1+1+2\) covariant scalars fully describing the LRS class of spacetimes are \cite{cc1}
\begin{eqnarray*}
\lbrace{A,\Theta,\phi, \Sigma, \mathcal{E}, \mathcal{H}, \rho, p, \Pi, Q, \Omega, \xi\rbrace}. 
\end{eqnarray*}
The quantity \(\phi\equiv\delta_ae^a\) is the sheet expansion, \(\Sigma\equiv\sigma_{ab}e^ae^b\) is the scalar associated with the shear tensor \(\sigma_{ab}\), \(\mathcal{E}\equiv E_{ab}e^ae^b\) is the scalar associated with the electric part of the Weyl tensor \(E_{ab}\), \(\mathcal{H}\equiv H_{ab}e^ae^b\) is the scalar associated with the magnetic part of the Weyl tensor \(\mathcal{H}_{ab}\), \(\Pi\equiv\pi_{ab}e^ae^b\) is the anisotropic stress scalar, \(Q\equiv -e^aT_{ab}u^b=q_ae^a\) is the scalar associated to the heat flux vector \(q_a\). The quantities \(\xi\) and \(\Omega\) are the spatial twist and rotation scalar respectively, which are defined by \(\xi=\left(1/2\right)\varepsilon^{ab}\delta_ae_b\) (where \(\varepsilon_{ab}\equiv \varepsilon_{abc}e^c= u^d\eta_{dabcd}e^c\) is the levi civita \(2\)-tensor, the volume element of the \(2\)-surface) and \(\Omega=e^a\omega_a\) (where \(\omega^a=\Omega e^a + \Omega^a\) is the rotation vector, with \(\Omega^a\) being the sheet component of \(\omega^a\)).

The full covariant derivatives of the vector fields \(u^a\) and \(e^a\) are given by \cite{cc1}
\begin{subequations}\label{4}
\begin{align}
\nabla_au_b&=-\mathcal{A}u_ae_b + e_ae_b\left(\frac{1}{3}\Theta + \Sigma\right) + N_{ab}\left(\frac{1}{3}\Theta -\frac{1}{2}\Sigma\right)\notag\\
&+\Omega\varepsilon_{ab},\label{4}\\
\nabla_ae_b&=-\mathcal{A}u_au_b + \left(\frac{1}{3}\Theta + \Sigma\right)e_au_b +\frac{1}{2}\phi N_{ab}\notag\\
&+\xi\varepsilon_{ab}.\label{444}
\end{align}
\end{subequations}
We also note the useful expression 
\begin{eqnarray}\label{redpen}
\hat{u}^a&=&\left(\frac{1}{3}\Theta+\Sigma\right)e^a.
\end{eqnarray}

The field equations for LRS spacetimes (we are interested in the case with vanishing cosmological constant) are given as propagation and evolution of the covariant scalars \cite{cc1}:

\begin{itemize}

\item \textit{Evolution}
\begin{subequations}
\begin{align}
\frac{2}{3}\dot{\Theta}-\dot{\Sigma}&=\mathcal{A}\phi- \frac{1}{2}\left(\frac{2}{3}\Theta-\Sigma\right)^2 - 2\Omega^2 + \mathcal{E} \notag\\
&- \frac{1}{2}\Pi - \frac{1}{3}\left(\rho+3p\right),\label{sube1}\\
\dot{\phi}&=\left(\frac{2}{3}\Theta-\Sigma\right)\left(\mathcal{A}-\frac{1}{2}\phi\right) + 2\xi\Omega \notag\\
&+ Q,\label{sube2}\\
\dot{\xi}&=-\frac{1}{2}\left(\frac{2}{3}\Theta-\Sigma\right)\xi \notag\\
&+ \left(\mathcal{A}-\frac{1}{2}\phi\right)\Omega,\label{sube3}\\
\dot{\Omega}&=\mathcal{A}\xi-\left(\frac{2}{3}\Theta-\Sigma\right)\Omega,\label{sube4}\\
\dot{\mathcal{H}}&=-3\xi\mathcal{E}-\frac{3}{2}\left(\frac{2}{3}\Theta-\Sigma\right)\mathcal{H}\notag\\
&+\Omega Q,\label{sube5}\\
\dot{\mathcal{E}}-\frac{1}{3}\dot{\rho}+\frac{1}{2}\dot{\Pi}&=-\left(\frac{2}{3}\Theta-\Sigma\right)\left(\frac{3}{2}\mathcal{E}+\frac{1}{4}\Pi\right)+\frac{1}{2}\phi Q\notag\\
&+3\xi\mathcal{H}+\frac{1}{2}\left(\frac{2}{3}\Theta-\Sigma\right)\left(\rho+p\right),\label{sube6}
\end{align}
\end{subequations}
\item \textit{Propagation}
\begin{subequations}
\begin{align}
\frac{2}{3}\hat{\Theta}-\hat{\Sigma}&=\frac{3}{2}\phi\Sigma + 2\xi\Omega + Q,\label{sube7}\\
\hat{\phi}&=-\frac{1}{2}\phi^2 + \left(\frac{1}{3}\Theta+\Sigma\right)\left(\frac{2}{3}\Theta-\Sigma\right)\notag\\
&+2\xi^2-\frac{2}{3}\rho-\mathcal{E}-\frac{1}{2}\Pi,\label{sube8}\\
\hat{\xi}&=-\phi\xi + \left(\frac{1}{3}\Theta+\Sigma\right)\Omega,\label{sube9}\\
\hat{\Omega}&=\left(\mathcal{A}-\phi\right)\Omega,\label{sube10}\\
\hat{\mathcal{H}}&=-\left(3\mathcal{E}+\rho+p-\frac{1}{2}\Pi\right)\Omega-3\phi\mathcal{H}\notag\\
&-Q\xi,\label{sube11}\\
\hat{\mathcal{E}}-\frac{1}{3}\hat{\rho}+\frac{1}{2}\hat{\Pi}&=-\frac{3}{2}\phi\left(\mathcal{E}+\frac{1}{2}\Pi\right)-\frac{1}{2}\left(\frac{2}{3}\Theta-\Sigma\right)Q\notag\\
&+3\Omega\mathcal{H}\label{sube12}
\end{align}
\end{subequations}
\item \textit{Evolution/Propagation}
\begin{subequations}
\begin{align}
\hat{\mathcal{A}}-\dot{\Theta}&=-\left(\mathcal{A}+\phi\right)\mathcal{A}-\frac{1}{3}\Theta^2+\frac{3}{2}\Sigma^2-2\Omega^2\notag\\
&+\frac{1}{2}\left(\rho+3p\right),\label{sube13}\\
\dot{\rho}+\hat{Q}&=-\Theta\left(\rho+p\right)-\left(2\mathcal{A}+\phi\right)Q\notag\\
&-\frac{3}{2}\Sigma\Pi,\label{sube14}\\
\dot{Q}+\hat{p}+\hat{\Pi}&=-\left(\mathcal{A}+\frac{3}{2}\phi\right)\Pi-\left(\frac{4}{3}\Theta+\Sigma\right)Q\notag\\
&-\left(\rho+p\right)\mathcal{A},\label{sube15}
\end{align}
\end{subequations}
\item \textit{Constraint}
\begin{eqnarray}\label{ggh10}
\mathcal{H}=3\Sigma\xi-\left(2\mathcal{A}-\phi\right)\Omega.
\end{eqnarray}

\end{itemize}

Furthermore, given a scalar \(\psi\) in LRS spacetimes, the dot and hat derivatives commute as follows \cite{cc1}:
\begin{eqnarray}\label{jun42}
\hat{\dot{\psi}}-\dot{\hat{\psi}}=-\mathcal{A}\dot{\psi}+\left(\frac{1}{3}\Theta+\Sigma\right)\hat{\psi}.
\end{eqnarray}

\subsection{Some definitions}
We define few useful notions. In the definitions and discussions that are to follow, \(k^a\) and \(l^a\) are respectively the outward and inward null normal vector fields to a leaf of such foliation defined by

\begin{eqnarray*}
k^a=\frac{1}{\sqrt{2}}\left(u^a+e^a\right)\text{  and  }l^a=\frac{1}{\sqrt{2}}\left(u^a-e^a\right),
\end{eqnarray*}
while \(\Theta_k\) and \(\Theta_l\) are the expansions of the congruences generated by \(k^a\) and \(l^a\) respectively.

\begin{definition}[\textbf{Trapped Surface}]\label{nggg}
A (future) trapped surface (TS) is a smooth, connected, closed, spacelike co-dimension \(2\) submanifold \(S\) of \(M\) such that the divergences, \(\Theta_k\) and \(\Theta_l\), of the congruences generated by the null normal vector fields \(k^a\) and \(l^a\) (\(k^a\) is the outgoing null normal vector field and \(l^a\) is the ingoing null normal vector field) respectively are everywhere negative on \(S\).  
\end{definition} 

\begin{definition}[\textbf{Marginally Trapped Surface}]\label{nggg0}
A marginally trapped surface (MTS) is a smooth, connected, closed, spacelike co-dimension \(2\) submanifold \(S\) of \(M\) such that \(\Theta_k\) is everywhere vanishing on \(S\) and \(\Theta_l\) is everywhere negative on \(S\). 
\end{definition}
In most cases, the definition of a marginally trapped surfaces is simply taken to be the case where \(\Theta_k\), vanishes.

\begin{definition}[\textbf{Marginally Trapped Tube}]\label{nggg1}
A marginally trapped tube (MTT) is a co-dimension 1 submanifold \(H\) of \(M\) which is foliated by MTS. 
\end{definition}
For more on the above definitions see the following references \cite{ib1,ash1,ash2,ash3}. In general, the signature of the induced metric on \(H\) will vary over \(H\). There are however cases where the signature is fixed all over \(H\). In such cases a spacelike marginally trapped tube is called a dynamical horizon (DH), a timelike marginally trapped tube is called a timelike membrane (TLM), and a null and non-expanding marginally trapped tube is called an isolated horizon (IH). 

\section{Result}\label{3}

We first recall some needed results for the proof of the main result of this paper.

For the LRS spacetimes, the outgoing null expansion, whose vanishing necessitates trapping, has been calculated as \cite{shef1}

\begin{eqnarray}\label{piapia1}
\Theta_k&=&\frac{1}{\sqrt{2}}\left(\frac{2}{3}\Theta-\Sigma+\phi\right).
\end{eqnarray}
So on any marginally trapped surface foliating the horizon one has \(\Theta_k=0\), i.e  

\begin{eqnarray}\label{jun41}
\frac{2}{3}\Theta - \Sigma + \phi=0.
\end{eqnarray}

Whether a horizon \(H\) is a timelike membrane, a dynamical horizon or an isolated horizon (or in general the causal character of an marginally trapped tube) can be determined by the function (see the references \cite{ib3,shef1} for more details)
\begin{eqnarray}\label{jun61}
C=\frac{\mathcal{L}_k\Theta_k}{ \mathcal{L}_l\Theta_k},
\end{eqnarray}
which is constant, where the operator \(\mathcal{L}_n\) denotes the Lie derivative along the vector field \(n^a\). Note that in general, the sign of \(C\) may vary over the marginally trapped tube, and as such the tube will have different components with different causal character (transitioning from spacelike to timelike of vice versa goes through a null component). However, there are cases where the sign of \(C\) remains fixed over \(H\). In such a case the MTT is called a dynamical horizon for \(C>0\), a timelike membrane and an isolated horizon if \(C=0\). For LRS spacetimes \eqref{jun61}, we calculate \(C\) explicitly as
\begin{eqnarray}\label{jun62}
C=\frac{2\left(\xi^2+2\xi\Omega-\Omega^2\right)-\left(\rho+p+\Pi\right)+2Q}{-2\left(\Omega^2+\xi^2\right)+\frac{1}{3}\left(\rho-3p\right)+2\mathcal{E}},
\end{eqnarray}
where we have noted \(\mathcal{L}_k\Theta_k=\dot{\Theta}_k+\hat{\Theta}_k\) and \(\mathcal{L}_l\Theta_k=\dot{\Theta}_k-\hat{\Theta}_k\), and used the equations \eqref{sube1}, \eqref{sube2}, \eqref{sube7} and \eqref{sube8} to calculate the dot and hat derivatives of the outgoing null expansion scalar \(\Theta_k\). 

The ingoing null expansion scalar was also calculated in \cite{shef1} as

\begin{eqnarray}\label{micc1}
\Theta_{l}=\frac{1}{\sqrt{2}}\left(\frac{2}{3}\Theta-\Sigma-\phi\right).
\end{eqnarray}

We now provide the following definition:

\begin{definition}\label{deffe1}
Let \(M\) be a spacetime manifold. A \(2\)-surface \(S\) in \(M\) is called minimally trapped if both the outgoing null expansion scalar \(\Theta_k\) and the ingoing null expansion scalar \(\Theta_l\) vanish everywhere on \(S\). An marginally trapped tube \(H\) will be called minimal if it is foliated by minimally trapped \(2\)-surfaces.
\end{definition}

\subsection{Main result}

We now state and prove the result of this paper.

\begin{theorem}\label{proo1}
An LRS spacetime in which at least one of the rotation or spatial twist is non-vanishing, cannot admit a non-minimal marginally trapped tube of the form \(\chi=X(t)\).
\end{theorem}
\begin{proof}
It was shown in \cite{say1} that, for any scalar \(\psi\) in LRS spacetime, the dot and hat derivatives satisfy the relation

\begin{eqnarray}\label{joni8}
\Omega\dot{\psi}=\xi\hat{\psi}.
\end{eqnarray}
This relation was used to obtain the following explicit algebraic expressions for \(\phi\) and \(Q\), obtained through the employment of the commutation relation \eqref{jun42}:

\begin{subequations}
\begin{align}
\phi&=\frac{\Omega}{\xi}\left(\frac{2}{3}\Theta-\Sigma\right),\label{joni9}\\ 
Q&=-\frac{\Omega\xi}{\left(\Omega^2+\xi^2\right)}\left(\rho+p+\Pi\right).\label{joni10}
\end{align}
\end{subequations}
Let \(H\) be an marginally trapped tube in \(M\) foliated by marginally trapped surfaces. We will show that if either one of \(\Omega\) and \(\xi\) is non-zero on \(H\), then either \(\left(2/3\right)\Theta-\Sigma=0\implies\phi=0\) or \(\phi=0\implies\left(2/3\right)\Theta-\Sigma=0\), in which case \(H\) is  minimal. 

Suppose we have that \(\Omega=0\) and \(\xi\neq 0\). Then, from \eqref{joni9} we have that \(\phi=0\), in which case we must have \(\left(2/3\right)\Theta-\Sigma=0\). Hence, \(H\) is minimal. On the other hand, if \(\Omega\neq0\) and \(\xi=0\), then \eqref{joni9} gives

\begin{eqnarray*}
\Omega\left(\frac{2}{3}\Theta-\Sigma\right)=0.
\end{eqnarray*}
Since \(\Omega\neq 0\), we must have \(\left(2/3\right)\Theta-\Sigma=0\), in which case \(\phi=0\). Hence \(H\) is also minimal.

Now, consider the case in which \(\Omega\neq0\) and \(\xi\neq0\). Then, on the horizon \eqref{joni9} becomes (using \eqref{jun41})

\begin{eqnarray}\label{jeje1}
\phi\left(\frac{\Omega}{\xi}+1\right)=0.
\end{eqnarray}
We know that \(\phi\neq 0\), since otherwise \(H\) would be minimal. Therefore we must instead have that

\begin{eqnarray}\label{jeje2}
\frac{\Omega}{\xi}=-1.
\end{eqnarray}

We also know that, in LRS spacetimes with simultaneous rotation and spatial twist, the existence of homothetic Killing vectors in the plane spanned by the unit tangent direction and the preferred spatial direction is guaranteed by Theorem \(1.\) of \cite{say1}. Therefore \(H\) also admits a homothetic Killing vector. The condition for \(H\) to be null is that \(\Omega/\xi=\pm 1\), clearly satisfied by \(H\) according to \eqref{jeje2}, which implies that \(H\) is a Killing horizon (see the reference \cite{say1} for more discussion). However, since \(H\) is null, one must have \(C=0\), which, from \eqref{jun62} gives

\begin{eqnarray}\label{jeje3}
2\left(\xi^2+2\xi\Omega-\Omega^2\right)-\left(\rho+p+\Pi\right)+2Q=0
\end{eqnarray}
Substituting \eqref{jeje2} into \eqref{jeje3} and \eqref{joni10} and simplifying, \eqref{jeje3} reduces to \(\xi^2=0\implies\xi=0\) (or \(\Omega^2=0\implies\Omega=0\)). The cases \(\xi=0\) and \(\Omega=0\) have already been shown to yield a minimal \(H\). The result therefore follows.\qed
\end{proof}

\section{Discussion}\label{4}

The result of Theorem \ref{proo1} severely constrains the spacetimes with locally rotational symmetry that can admit dynamical horizons and timelike membranes. In particular, Theorem \ref{proo1} tells us that, for spacetimes with locally rotational symmetry where all vector and tensor quantities vanish identically, rotation and spatial twist are obstacles to the existence of truly dynamical black holes. Hence, only the class II of these spacetimes with \(\Omega=0\) and \(\xi=0\) admit such horizons. Well known examples of this class include the Lemaitre-Tolman-Bondi (LTB) solution,  Schwarzschild spacetime (although the MTT here, a null horizon, is minimal since both expansion scalars are identically zero), the Oppenheimer-Snyder (OS) dust solution and the Robertson-Walker (RW) solution. (The LTB and RW solutions, for example, contain dynamical horizons foliated by non-minimal MOTS, and the OS dust collapse contains timelike membranes, also foliated by nonminimal MOTS. (These were recently studied in context of the 1+1+2 formulation \cite{shef1}.)

While this itself does not expound on physical application of the obtained result, it certainly adds to the growing literature on the existence problem for dynamical horizons in particular. 

In the numerical context, the coordinate dependence of \eqref{jan29191} already restricts the possible solutions in the LRS class of spacetimes that could be useful in modelling mergers. The result of this work appears to verify this, and rules out many LRS solutions which could possibly be used to model  black hole mergers. Two important subclasses are the spatially homogeneous LRS III class of spacetimes with \(\xi\neq0\) and \(\Omega\) and the LRS I solutions with \(\xi=0\) and \(\Omega\neq0\). A known example of the latter is the G\"{o}del rotating solution.

An immediate problem to consider would be the rigidity of the result of this paper, under perturbations, where one has an LRS background.

\section*{Acknowledgements}
Gratitude is expressed to the anonymous referee who provided invaluable comments and suggestions that have enhanced the clarity of the result of this paper. The author acknowledges support by the First Rand Bank, through the Department of Mathematics and Applied Mathematics, University of Cape Town, South Africa.

\end{document}